\documentclass[12pt]{article}
\usepackage{graphicx} 

\usepackage{amsmath}
\usepackage{cleveref}
\usepackage{color}
\usepackage{xspace}
\usepackage{tikz}
\usepackage{amsthm}
\usepackage{amssymb}

\usepackage{enumitem}
\setlist{nosep}

\newcommand{\imped}{impedensable\xspace}
\newcommand{\Imped}{Impedensable\xspace}

\newcommand{\snote}[1]{{\color{black} #1}}
\definecolor{darkgreen}{rgb}{0.0,0.5,0.0}
\newcommand{\arya}[1]{{\color{black} #1}}

\newtheorem{definition}{Definition}
\newtheorem{algorithm}{Algorithm}
\newtheorem{example}{Example}
\newtheorem{theorem}{Theorem}

\newtheorem{lemma}{Lemma}

\title{Tolerance to Asynchrony in Algorithms for Multiplication and Modulo
\thanks{\textbf{To appear in Theoretical Computer Science.}}
\thanks{\textbf{A preliminary version of
this article
appeared in the 25th International Symposium on Stabilization, Safety, and Security of Distributed Systems (SSS 2023) \cite{Gupta2023}.
}}
\thanks{\textbf{This is an extended abstract of the SSS 2023 paper. To get the arXiv version of the SSS paper, please go to the immediately previous version uploaded on this arXiv repository (2302.07207v4, uploaded on 24 Jul 2023). This extended abstract is uploaded to this same arXiv repository due to the policies of arXiv on submissions with shared content.
}}
}
\author{Arya Tanmay Gupta and Sandeep S Kulkarni}
\date{Computer Science and Engineering, Michigan State University\\ \texttt{\{atgupta,sandeep\}@msu.edu}}

\begin{document}

\maketitle

\begin{abstract}

    In this article, we study some parallel processing algorithms for multiplication and modulo operations. We demonstrate that the state transitions that are formed under these algorithms satisfy lattice-linearity, where these algorithms induce a lattice among the global states. 
    Lattice-linearity implies that these algorithms can be implemented in asynchronous environments, where the nodes are allowed to read old information from each other. It means that these algorithms are guaranteed to converge correctly without any synchronization overhead. 
    These algorithms also exhibit snap-stabilizing properties, i.e., starting from an arbitrary state, the sequence of state transitions made by the system strictly follows its specification.

\end{abstract}


\section{Introduction}

The development of parallel processing algorithms is increasingly gaining interest. This is because computing machines are manufactured with multiprocessor chips as we face a bound on the rate of architectural development of individual microprocessors. Such algorithms, in general, require synchronization among their processing nodes. Without such synchronization, the nodes perform executions based on inconsistent values possibly resulting in substantial delay or potentially incorrect computation.

Consider the problem of graph colouring for example. An algorithm for this problem can be designed as follows. In each step, a node reads the colour of all its neighbours and, if necessary, updates its own colour. Let us consider that we allow this algorithm to execute in asynchrony: consider that there are two nodes, $i$ and $j$, in the input graph, connected to each other by an edge, and they execute asynchronously. Let us suppose that $i$ reads the state of $j$ and changes its state. The action of $i$ is acceptable till now. 
However, if $j$ reads the value of $i$ concurrently, it is relying on inconsistent information about $i$. This is because $i$ will have changed its state in its immediate next step. In other words, $j$ is relying on old information about $i$. Such behaviour is prevented by the use of apt synchronization primitives. However, they introduce considerable overhead in terms of time and resources.

If an algorithm guarantees correctness without synchronization, then the overheads of the otherwise required time and resources can be eliminated. In the context of algorithms that allow asynchrony, lattice theory has provided very useful concepts. Lattice-linear systems induce a lattice structure in the state space, which allows the nodes to perform executions asynchronously. 
{
Lattice-linearity was utilized in modelling the problems (called \textit{lattice-linear problems} \cite{Garg2020}), and in developing algorithms (called \textit{lattice-linear algorithms}).
In a lattice-linear problem, nodes in violating local states can be determined within the natural constraints of the problem.
Lattice-linear algorithms impose a lattice structure in problems, which may or may not be lattice-linear \cite{Gupta2021,Gupta2022}.
}
Thus, the algorithms that traverse such a state transition system can allow the nodes to perform executions asynchronously while preserving correctness. 

Lattice-linearity allows inducing single or multiple lattices among the global states. If self-stabilization is required, then for every induced lattice, its supremum must be an optimal state. This way, it can be ensured that the system can traverse to an optimal state from an arbitrary state.
We introduced eventually lattice-linear algorithms \cite{Gupta2021} and fully lattice-linear algorithms \cite{Gupta2022} for non-lattice-linear problems (we discuss more on this in \Cref{section:preleminaries} for the sake of completeness, but non-lattice-linear problems are not the focus of this article). These algorithms are self-stabilizing.

Many lattice-linear problems 
do not allow self-stabilization \cite{Garg2020,Garg2021,Garg2022}.
While developing new algorithms for lattice-linear and non-lattice-linear problems is very interesting, it is also worthwhile to study if algorithms already present in the literature exploit lattice-linearity of problems or if lattice-linearity is present in existing algorithms.
For example, lattice-linearity was found to be exploited by Johnson's algorithm for shortest paths in graphs \cite{Garg2020} and by Gale’s top trading cycle algorithm for {the housing market problem} \cite{Garg2021}.

In this article, we study parallel implementations of two fundamental operations in mathematics: multiplication and modulo. We study algorithms for $n\times m$ and $n \mod m$ where $n$ and $m$ are large integers represented as binary strings.

The applications of integer multiplication include the computation of power, matrix products which has applications in a plethora of fields including artificial intelligence and game theory, the sum of fractions and coprime base.
Modular arithmetic has applications in theoretical mathematics, where it is heavily used in number theory and various topics (groups, rings, fields, knots, etc.) in abstract algebra. 
Modulo also has applications in applied mathematics, where it is used in computer algebra, cryptography, chemistry and the visual and musical arts. In many of these applications, the value of the divisor is fixed.

A crucial observation is that these operations are lattice-linear. Also, the algorithms that we study in this article are capable of computing the correct answer even if the nodes are initialized in an arbitrary state, i.e., the algorithms that we study in this article are self-stabilizing. These properties are present in these operations as opposed to many other lattice-linear problems where the lattice structure does not allow self-stabilization.

\subsection{Contributions of the article }

\begin{itemize}
    \item We study the lattice-linearity of modulo and multiplication operations.
    The design and woking of the algorithms imply that these problems satisfy the contstraints of lattice-linear problems (lattice linear problems are studied in \cite{Garg2020}, however, they were first formally classified (against non-lattice-linear problems) in \cite{Gupta2023a}). 
    \item We also show that these algorithms exhibit properties that are similar to snap-stabilizing algorithms \cite{Bui1999}, i.e., starting from an arbitrary state, the sequence of state transitions made by the system strictly follows its specification, i.e., initializing in an arbitrary state, the nodes immediately start obtaining the values as expected with each action they execute.
    \item We observe that both multiplication and modulo manifest multiple lattice structures that can be exploited with different numbers of computing nodes. It means that depending on the number of processors available, we can choose which algorithm to use. Modulo and multiplication are the first lattice-linear problems for which multiple lattice structures are studied.
\end{itemize}

\subsection{Organization of the article }

\noindent In \Cref{section:preleminaries}, we describe the definitions and notations that we use in this article. We study lattice-linearity of the multiplication operation in \Cref{section:mul-parallel}.
Then, in \Cref{section:mod-parallel}, we study the lattice-linearity of the modulo operation. 
\Cref{section:discussion} discusses the properties commonly present in all the algorithms that we study in this article.
In \Cref{section:literature}, we discuss the related work.
Finally, we conclude in \Cref{section:conclusion}.

\section{Preliminaries}\label{section:preleminaries}

This article focuses on multiplication and modulo operations, where the operands are $n$ and $m$. In the computation $n\times m$ or $n\mod m$, $n$ and $m$ are the values of these numbers respectively, and $|n|$ and $|m|$ are the length of the bitstrings required to represent $n$ and $m$ respectively.
\footnote{{Since $n$ and $m$ are sequence of bit-values and if $x$ is a sequence or a set, $|x|$ is used to denote the number of elements in $x$, $|n|$ and $|m|$ are the lengths of these bitstrings respectively. This notation does not represent the magnitude of their values.}}
If $n$ is a bitstring, then $n[k]$ is the $k^{th}$ bit of $n$ (indices start from 1). For a bitstring $n$, $n[1]$ is the most significant bit of $n$ and $n[|n|]$ is the least significant bit of $n$. We use $n[j:k]$ to denote the bitstring from $j^{th}$ bit to $k^{th}$ bit of $n$; this includes $n[i]$ and $n[j]$. For simplicity, we stipulate that $n$ and $m$ are of lengths in some power of $2$.
Since size of $n$ and $m$ may be substantially different, we provide complexity results that are of the form $O(f(n,m))$ in all cases, where $f$ is a function of $n$ and $m$.

We use the following string operations: (1) $append(a$, $b)$, appends $b$ to the end of $a$ in $O(|b|)$ time, (2) $rshift(a, k)$, deletes rightmost $k$ bits of $a$ in $O(k)$ time, and (3)  $lshift(a, k)$, appends $k$ zeros to the right of $a$ in $O(k)$ time. 

{
In several places, we have used the functions $\textsc{Mod}(x,y)$, $\textsc{Mul}(x,y)$, and $\textsc{Sum}(x,y)$. These functions, respectively, compute $x\mod y, x\times y$ and $x+y$.
}

\subsection{Modeling Distributed Algorithms}

$n \times m$ or $n \mod m$ are typically thought of as arithmetic operations. However,  when $n$ and $m$ are large, we view them as algorithms. In this article, we view them as parallel/distributed algorithms where the nodes collectively perform computations to converge to final output.

A parallel/distributed algorithm consists of multiple nodes where each node is associated with a set of variables. A \textit{global state}, say $s$, is obtained by assigning each variable of each node a value from its respective domain. $s$ is represented as a vector, where $s[i]$ itself is a vector of the variables of node $i$.
$S$ denotes the \textit{state space}, which is the set of all global states that a given system can obtain. 

Each node is associated with rules. Each rule at node $i$ checks the values of other nodes and updates its own variables. An \textit{rule} is of the form $g\longrightarrow a_c$, where $g$ is the \textit{guard} (a Boolean predicate involving variables of $i$ and other nodes) and \textit{action} $a_c$ is an instruction that updates the variables of $i$.
We assume all actions to be executed atomically.
{If at least one guard of a node is true, then that node is \textit{enabled}, otherwise it is \textit{disabled}.}

An algorithm $A$ is \textit{self-stabilizing} with respect to the subset $S_o$ of $S$ where $S$ is the set of all global states, iff (1) \textit{convergence:} starting from an arbitrary state, any computation of $A$ reaches a state in $S_o$, and (2) \textit{closure:} any computation of $A$ starting from $S_o$ always stays in $S_o$. 
We assume $S_o$ to be the set of optimal states: the system is deemed converged once it reaches a state in $S_o$. 
{$A$ is \textit{silent} if all guards are false once the system reaches a state in $S_o$, and so no node makes a move.}
An algorithm is \textit{snap-stabilizing} iff starting from an arbitrary state, it makes the system follow a sequence of state transitions as per the specification of that system.

\snote{
The term \textit{time complexity} defines the time taken by an algorithm to reach an optimal state. 
\textit{Work complexity} refers to the total work done by the algorithm, i.e., work complexity is the time it would have taken if all the work in the algorithm was done sequentially.
Since a parallel algorithm can complete multiple tasks at the same time, time complexity can be lower than work complexity. 
As an example, if an algorithm executes actions 1, 2, and 3, each taking 1 time unit, such that actions 1 and 2 are executed concurrently and action 3 is executed only after they finish, then the time complexity is 2 units and the work complexity is 3 units. 
}

{
In several algorithms, a tree is induced among the computing nodes due to the flow of data. In such a tree, we denote $level$ to be the marker of the distance of the nodes from the root node. E.g., the root node is the topmost node, so it is at the highest level, and the nodes that are farthest from the root node are at level 1.
}

\subsection{Execution without Synchronization}
Typically, we view the computation of an algorithm as a sequence of global states $\langle s_0, s_1, \cdots\rangle$ where $s_{t+1}, t\geq 0,$ is obtained by executing some action by one or more nodes in $s_t$. 
Under proper synchronization, node $i$ evaluates its guards on the \textit{current} local states of its neighbours. 

For the sake of discussion, assume that only node $i$ is executing in state $s_t$. Let  $s_{t+1}$ be the state obtained by executing $g \longrightarrow a_c$ (of node $i$) in $s_t$, under the required synchronization.
To understand how execution without synchronization works, let $x(s)$ denote the value of some variable $x$ in state $s$.
If $g$ is evaluated without synchronization, then node $i$ may read old values of some variables.
Let the computation prefix til $s_t$ be $\langle s_0, s_1, \cdots, s_t\rangle$. 
Then, the state may appear to $i$ as $s'\neq s_t$, where some values are old. Specifically, the state $s_t$ will appear as $s'$ where $x(s')\in\{ x(s_0), x(s_1), \cdots, x(s_t) \}.$ If $g$ evaluates to false, then the next state will stay at $s_t$, which may be incorrect. If, otherwise, $g$ evaluates to true then node $i$ will execute $a_c$, and the resulting state $s''$ may not be the same as $s_{t+1}$, where (1) the state of $i$, $s''[i]$, is evaluated on the values of variables in $s'$ (and not $s_t$), and (2) $\forall j\neq i:$ the state of $j$ will be unchanged, i.e., $s''[j] = s_t[j]$.

The model described in the above paragraph is \textit{arbitrary asynchrony}, a model in which a node can read old values of other nodes arbitrarily, requiring that if some information is sent from a node, it eventually reaches the target node.
In this article, however, we are interested in \textit{asynchrony with monotonous read} (AMR) model. 
AMR model is arbitrary asynchrony with an additional restriction: when node $i$ reads the state of node $j$, the reads are monotonic, i.e., if $i$ reads a newer value of the state of $j$ then it cannot read an older value of $j$ at a later time. For example, if the state of $j$ changes from $0$ to $1$ to $2$ and node $i$ reads the state of $j$ to be $1$ then its subsequent read will either return $1$ or $2$, it cannot return $0$.  

\subsection{Embedding a $\prec$-lattice in Global States}

Let $s$ denote a global state, and let $s[i]$ denote the state of node $i$ in $s$. First, we define a total order $\prec_l$; all local states of a node $i$ are totally ordered under $\prec_l$. 
Using $\prec_l$, we define a partial order $\prec_g$ among global states as follows: 

We say that $s \prec_g s^\prime$ iff $(\forall i: s[i]=s'[i]\lor s[i]\prec_l s'[i]) \land (\exists i:s[i]\prec_ls'[i])$.
Also, $s=s'$ iff $\forall i:s[i] = s'[i]$. 
For brevity, we use $\prec$ to denote $\prec_l$ and $\prec_g$: $\prec$ corresponds to $\prec_l$ while comparing local states, and $\prec$ corresponds to $\prec_g$ while comparing global states. 
We also use the symbol `$\succ$' which is such that $s\succ s'$ iff $s' \prec s$.
Similarly, we use symbols `$\preceq$' and `$\succeq$'; e.g., $s\preceq s'$ iff  $s=s' \lor s \prec s'$.
We call the lattice, formed from such partial order, a \textit{$\prec$-lattice}.

\begin{definition}\label{definition:<-lattice}
    \textbf{{\boldmath$\prec$}-\textit{lattice}}. 
    Given a total relation $\prec_l$ that orders the states visited by $i$ (for each $i$) the $\prec$-lattice corresponding to $\prec_l$ is defined by the following partial order:
    $s \prec s'$ iff $(\forall i: s[i] \preceq_l s'[i]) \wedge (\exists i: s[i] \prec_l s'[i])$.
\end{definition}

A $\prec$-lattice constraints how global states can transition among one another: state $s$ can transition to state $s'$ iff $s\prec s'$. In the $\prec$-lattice discussed above, we can define the meet and join of two states in the standard way: the meet (respectively, join), of two states $s_1$ and $s_2$ is a state $s_3$ where $\forall i: s_3[i]$ is equal to $min(s_1[i], s_2[i])$ (respectively, $max(s_1[i], s_2[i])$), 
{where for a pair of local states $x$ and $y$, if $x\prec_l y$, then $\min(x, y) = \min(y, x)=x$ and 
$\max(x, y) = \max(y, x)=y$.
}
For $s_1$ and $s_2$, their meet (respectively, join) has a path to (respectively, is reachable from) both $s_1$ and $s_2$.

By varying $\prec_l$ that identifies a total order among the states of a node, one can obtain different lattices. A $\prec$-lattice, embedded in the state space, is useful for permitting the algorithm to execute asynchronously.
Under proper constraints on the structure of $\prec$-lattice, convergence can be ensured. 

\subsection{Lattice-Linear Problems}
Next, we discuss \textit{lattice-linear problems}, i.e., the problems where the problem statement creates the lattice structure automatically. Specifically, in such problems, in any given suboptimal global state, all the nodes that are in a bad local state can be identified; such nodes must change their state in order for the system to reach an optimal state. Lattice-linear problems can be represented by a predicate that induces a lattice among the states in $S$. 

In \textit{lattice-linear problems},
a problem $P$ can be represented by a predicate $\mathcal{P}$ such that for any node $i$, if it is violating $\mathcal{P}$ in some state $s$, then it must change its state. Otherwise, the system will not satisfy $\mathcal{P}$.
Let $\mathcal{P}(s)$ be true iff state $s$ satisfies $\mathcal{P}$. A node violating $\mathcal{P}$ in $s$ is called an \textit{\imped} node
\footnote{\arya{We use the term `impedensable' as a combination of the English words impediment and indispensable.
The term `\imped' is similar to the notion of a node being \textit{forbidden} introduced in \cite{Garg2020}. This word itself comes from predicate detection background \cite{Chase1995}. 
We changed the notation to avoid the misinterpretation of the English meaning of the word `forbidden'.}}
(an \textit{impediment} to progress if does not execute, \textit{indispensable} to execute for progress). Formally,

\begin{definition}\label{definition:impedensable-node}\cite{Garg2020} \textbf{\textit{\Imped node.}} $\textsc{\Imped}(i,s,\mathcal{P})\equiv \lnot \mathcal{P}(s)$ $\land$ $(\forall s'\succ s:s'[i]=s[i]\Rightarrow\lnot \mathcal{P}(s'))$. \end{definition}

\Cref{definition:impedensable-node} implies that if a node $i$ is \imped in state $s$, then in any state $s'$ such that $s'\succ s$, if the state of $i$ remains the same, then the algorithm will not converge. By definition, a node should also not come back to a local state that it already rejected, {because that rejected state was not viable for any optimal global state, so otherwise, an optimal global state cannot be achieved}.
Thus, $\mathcal{P}$ induces a total order among the local states visited by a node, for all nodes. Consequently, the discrete structure that gets induced among the global states is a $\prec$-lattice, as described in \Cref{definition:<-lattice}. 
Thus, any $\prec$-lattice among the global states is induced by a predicate $\mathcal{P}$ that satisfies \Cref{definition:impedensable-node}.

There can be multiple arbitrary lattices that can be induced among the global states {by changing several aspects of a system, e.g., by changing the predicate}. A system cannot guarantee convergence while traversing an arbitrary lattice. To resolve this, we design the predicate $\mathcal{P}$ such that {it follows \Cref{definition:impedensable-node}}.
Thus, $\mathcal{P}$ induces a $\prec_l$ relation among the local states, and as a result, a $\prec$-lattice among the global states.
We say that $\mathcal{P}$ is \textit{lattice-linear} with respect to that $\prec$-lattice.
Consequently, in any suboptimal global state, there will be at least one \imped node. Formally,

\begin{definition}\cite{Garg2020}\textbf{\textit{Lattice-Linear Predicate.}}
    $\mathcal{P}$ is a lattice-linear predicate with respect to a $\prec$-lattice induced among the global states iff $\forall s\in S: \lnot\mathcal{P}(s) \Rightarrow \exists i:\textsc{\Imped}(i,s,\mathcal{P})$.
\end{definition}

Now we complete the definition of lattice-linear problems. In a lattice-linear problem $P$, given any suboptimal global state, we can identify all and only nodes which cannot retain their state{, i.e., nodes which are \imped}. In this article, we observe that the algorithms that we study exploit this nature of their respective problems. $\mathcal{P}$ is thus designed conserving this nature of the subject problem $P$.

\begin{definition}\label{definition:ll-problem}
    \textbf{Lattice-linear problems}.
    A problem $P$ is lattice-linear 
    iff there exists a predicate $\mathcal{P}$ and a $\prec$-lattice such that
    
    \begin{itemize}
        \item $P$ is deemed solved iff the system reaches a state where $\mathcal{P}$ is true,
        \item $\mathcal{P}$ is lattice-linear with respect to the $\prec$-lattice induced in $S$, i.e., $\forall s: \neg \mathcal{P}(s) \Rightarrow \exists i:\textsc{\Imped}(i,s,\mathcal{P})$.
        \item $\forall s:(\forall i:\textsc{\Imped}(i,s,\mathcal{P})\Rightarrow (\forall s':\mathcal{P}(s')\Rightarrow s'[i]\neq s[i]))$.
    \end{itemize}
\end{definition}

Problems such as stable marriage, job scheduling and market clearing price, as studied in \cite{Garg2020}, are lattice-linear problems.
In this article, we study lattice structures that can be induced in multiplication and modulo: we show that multiplication and modulo are lattice-linear problems. All the lattice structures that we study in this article allow self-stabilization, i.e., the supremum of the lattice induced in the state space is the optimal state.

\begin{definition}\label{definition:ssll-problem}
    \textbf{Self-stabilizing lattice-linear predicate}.
    Continuing from \Cref{definition:ll-problem},
    $\mathcal{P}$ is a self-stabilizing lattice-linear predicate if and only if the supremum of the lattice that $\mathcal{P}$ induces is an optimal state, i.e., $\mathcal{P}(supremum(S))=true$.
\end{definition}

\noindent Note that $\mathcal{P}$ can also be true in states other than the supremum of the $\prec$-lattice. 

\noindent\textbf{\textit{Remark:}} A $\prec$-lattice, induced under $\mathcal{P}$, allows asynchrony because if a node $i$, reading old values, reads the current state as $s$, then for the current state $s'$, $s\prec s'$. Thus, if $i$ evaluates that it is impedensable in $s$, then $\lnot\mathcal{P}(s)\Rightarrow \lnot\mathcal{P}(s')$ because $s[i]=s'[i]$.

\subsection{Modulo: some classic sequential models}\label{section:mod-sequential}

In this subsection, we discuss some sequential algorithms for computing $n \mod m$. We will utilize these preliminary algorithms to analyze the effective time complexity of parallelized modulo operation. We consider two instances, one where both $n$ and $m$ are inputs and another where $n$ is an input but $m$ is hardcoded. We utilize these algorithms in \Cref{section:parallel-modulo}.

The latter algorithm is motivated by algorithms such as RSA \cite{Rivest1978} where the value of $n$ changes based on the message to be encrypted/decrypted, but the value of $m$ is fixed once the keys are determined. Thus, some pre-processing can potentially improve the performance of the modulo operation;
we observe that certain optimizations are possible. While the time and space complexities required for preprocessing in this algorithm are high, thereby making it impractical, it demonstrates a gap between the lower and upper bound in the complexity. 

\subsubsection{Modulo by Long Division}\label{subsection:long-division}
~\\
The standard long-division algorithm to compute $n$ mod $m$ is shown below. 
\begin{center}
    \small 
    \noindent\begin{tabular}{|l l|}
        \hline 
         & \textsc{Division-Modulo}($n,m$):\\
        1. & \quad $ans = n[1:|m|-1]$. $k = |m|$.\\
        2. & \quad \textbf{while} $(k \leq |n|)$, \\
        3. & \quad \quad $ans=append(ans,n[k])$.\\
        4. & \quad \quad \textbf{if} $ans>m$, then\\
        5. & \quad \quad \quad $ans=ans-m$. k=k+1.\\
        6. & \quad Return $ans$.\\
        \hline 
    \end{tabular}
\end{center}

Clearly, the number of iterations in this algorithm is bounded by $|n|$. In each of these iterations, the worst case complexity is $O(|m|)$ to perform the subtraction operation. Thus, the time complexity of standard long division is $O(|n|\times |m|)$ when $m$ and $n$ are both inputs to the algorithm. 

\subsubsection{Modulo by constructing DFA }\label{subsection:dfa-modulo}
~\\
If the value of $m$ is hardcoded in the algorithm, we use it to reduce the cost of the modulo operation by creating a deterministic finite automaton (DFA) $M=\langle Q$, $\Sigma$, $\delta$, $q_0\rangle$, where (1) $Q=\{q_0..q_{m-1}\}$ is the set of all possible states of $M$, (2) $\Sigma=\{0,1\}$ is the alphabet set, (3) $\delta$ is the transition function, the details of which we study in this section, and (4) $q_0$ is the initial state. If $M$ has read the first $k$ digits of $n$ then its state would be $n[1:k] \mod m$. When $M$ reads the next digit of $n$, $n[1:(k+1)] \mod m$ is evaluated depending on the next digit. If the next digit is 0, the next state will be $(2\times (n[1:k] \mod m)) \mod m$, otherwise, it will be  $(2\times (n[1:k] \mod m)+1) \mod m$.
Since the value of $m$ is hardcoded, this DFA is assumed to be pre-computed. 
As an example, for $m=3$, the corresponding DFA is provided in \Cref{example:mod3}.

\begin{example}\label{example:mod3}
    A finite automaton $M_3$ computing $n\mod 3$ ($m=3$ is fixed) for any $n\in \mathbb{N}$ is shown in \Cref{figure:mod3-dfa}.
    \qed 
\end{example}

\begin{figure}[ht]
    \centering
    \begin{tikzpicture}
        \node [circle, draw=black] at (0,0) (q0) {$q_0$};
        \node [circle, draw=black] at (2,0) (q1) {$q_1$};
        \node [circle, draw=black] at (4,0) (q2) {$q_2$};
        
        \draw[->] (q0.north west) to [out=135,in=225,looseness=5] (q0.south west);
        
        \draw[->,thick] (0,1) -- (0,.5);
        
        \draw[->] (q0.north east) to [out=45, in=135] (q1.north west);
        \draw[->] (q1.south west) to [out=225, in=315] (q0.south east);
        
        
        \draw[->] (q1.north east) to [out=45, in=135] (q2.north west);
        \draw[->] (q2.south west) to [out=225, in=315] (q1.south east);
        
        \draw[->] (q2.north east) to [out=45,in=315,looseness=5] (q2.south east);
        
        \node at (1,1) {1};
        \node at (1,-1) {1};
        \node at (-1,0) {0};
        
        \node at (3,1) {0};
        \node at (3,-1) {0};
        
        \node at (5,0) {1};
    \end{tikzpicture}
    \caption{A finite automaton $M_3$ computing $n\mod 3$ for any $n\in \mathbb{N}$.}
    \label{figure:mod3-dfa}
\end{figure}

With this DFA, the cost of computing the modulo operation corresponds to one DFA transition for each digit of $n$. Hence, the complexity of the corresponding operation is $O(|n|)$.
$M$ does not have any accepting states, which is unlike a usual finite automaton; the final state of $M$ only tells us the value of the remainder which we would obtain after the computation of the expression $n\mod m$.

We define the construction of the transition function $\delta$ of $M$ in \Cref{figure:dfa-modulo}.
$\delta$ is constructed using the magnitude of $m$, and it is capable of reading a bitsting $n$ starting from $n[1]$ and traversing through $n[|n|]$, reading every bit, sequentially, in each step. The problem is as follows: let that $M$ has read the first $z$ bits of $n$ that evaluates to the value $K$ (here, $K=n[1:z]$), and let $K\mod m=k$, so $M$ is in state $q_k$. From here, we have to determine the next state based on whether the next bit $M$ reads is 0 or 1, which would mean that the total value read will be $2\times K+0$ or $2\times K+1$.

We start by assigning edges from $q_0$. $(2\times0+0) \mod m$ is 0 and $(2\times 0+1)\mod m$ is 1, for example. So $\delta(q_0,0)=0$ and $\delta(q_0,1)=1$. So we assign $q_0$ to transition to $q_0$ on input 0, and to $q_1$ on input 1. This method induces labelled edges between the states of $M$, such that those labels define which edge should be traversed based on what bit is read. After reading the first bit, assuming that $M$ has read the value $K$ so far and $M$ is in state $q_k$, we assign the next state to be $(2\times k+0)\mod m$ and $(2\times k+1)\mod m$ for inputs 0 and 1 respectively. This is because $(2\times K+0)\mod m$ and $(2\times K+1)\mod m$ will be equal to $(2\times k+0)\mod m$ and $(2\times k+1)\mod m$ respectively. Note that if $2\times k+0$ (resp., $2\times k+1$) is greater than $m$, then we assign $(2\times k+0)\mod m$ (resp., $(2\times k+1)\mod m$) to be $(2\times k+0)-m$ (resp., $(2\times k+1)-m$) as $2\times k+0$ (and $2\times k+1$) cannot exceed $2\times m$.

\begin{figure}[ht]
    \centering
    \begin{tabular}{|l l|}
        \hline 
         & Elaborated definition of $\delta$.\\
         & \\
        1. & $v=0$. $m_{rs}=rshift(m, 1)$.\\
         & \quad $b=$ last bit of $m$. $i=0$.\\
        2. & \textbf{for}( ; $i< m_{rs}$; $i=i+1$),\\
        3. & \quad $\delta(q_i,0)=q_v$. $v= v+1$.\\
        4. & \quad $\delta(q_i,1)=q_v$. $v= v+1$.\\
        5. & \textbf{if} $b=1$, then\\
        6. & \quad {$\delta(q_i,0)=q_v$. $\delta(q_i,1)=q_0$. $v=1$.}\\
        7. & \textbf{else}, then\\
        8. & \quad {$\delta(q_i,0)=q_0$. $\delta(q_i,1)=q_1$. $v=2$.}\\
        9. & $i=i+1$.\\
        10. & \textbf{for}( ; $i< m$; $i=i+1$),\\
        11. & \quad $\delta(q_i,0)=q_v$. $v= v+1$.\\
        12. & \quad $\delta(q_i,1)=q_v$. $v= v+1$.\\
        13. & {$\delta(q_i,0)=q_0$. $\delta(q_i,1)=q_1$.}\\
        \hline 
    \end{tabular}
    \caption{Definition of the transition function $\delta$.}
    \label{figure:dfa-modulo}
\end{figure}

It can be clearly observed that while this approach takes $|n|$ steps, each step taking a constant amount of time, the time complexity (as well as the space complexity) of the required preprocessing is $O(m\times |m|)$, which is very high. Therefore, this approach is not practical when $m$ is large. However, we consider it to observe that we obtain a time complexity of $O(|n|)$ to run this automaton to evaluate for the modulo operation if $m$ can be hardcoded. 

\section{Parallelized Multiplication Operation}\label{section:mul-parallel}

In this section, we demonstrate that the parallelized version of the standard multiplication algorithm as well as a parallelized version of Karatsuba's \cite{Karatsuba1962} algorithm presented in \cite{Cesari1996} meet the requirements of lattice-linearity, i.e. a system of nodes traverses a lattice of global states and provide the final output. 
We consider the problem where we want to compute $n\times m$.

\subsection{Parallelizing Standard Multiplication}\label{subsection:multiplicaiton-standard}

In this subsection, we present the parallelization of the standard multiplication algorithm. First, we discuss the key idea of the sequential algorithm, then we elaborate on the lattice-linearity of its parallelization.

\subsubsection{Key Idea}

In the standard multiplication, we multiply $m$ with one digit of $n$ at a time (for each digit of $n$), and then add all these multiplications, after left shifting those resultant strings appropriately. Suppose that we have two strings $a=n[1:\lfloor |n|/2\rfloor]\times m$ and $b=n[\lfloor |n|/2\rfloor+1:|n|]\times m$. Then the resultant multiplication $n\times m$ will be equal to $lshift(a,|n|-\lceil |n|/2\rceil)+b$.

\subsubsection{Parallelization}

This algorithm requires $2\times |n|-1$ nodes, and induces a binary tree among them. The root of the tree is marked as node $1$ and any node $i$ ($1\leq i\leq |n|-1$) has two children: node $2i$ and node $2i+1$.

In this algorithm, every node stores two variables: $shift$ and $ans$. We demonstrate that the computation of each of these variables is lattice-linear. \\~

\noindent \textbf{Computation of {\boldmath $i[shift]$}}: At the lowest level (level 1), the value of $shift$ is set to $0$. 
Consequently, at the next level (level 2), $shift$ is set to $1$. At all higher levels, $shift$ of any node is computed to be twice the value of $shift$ of its children, i.e., $i[shift]$ is set to $2\times (2i)[shift]$. This can be viewed as a lattice-linear computation where a node is \imped iff the following condition is satisfied. A \imped node updates its value to be either $0$ (at level 1), $1$ (at level 2), or $2\times (2i)[shift]$ (at level 3 and higher).  
{
$$\begin{array}{c}
\textsc{\Imped-Multiplication-Standard-Shift}(i)\equiv\\
    \begin{cases}
        i[shift]\neq 0 & \text{if $i\geq|n|$}\\
        i[shift]\neq 1 & \text{if $(2i)[shift]=(2i+1)[shift]=0$}\\
        i[shift]\neq 2\times (2i)[shift] & \text{if $(2i)[shift]=(2i+1)[shift]\geq 1$}
    \end{cases}
\end{array}
$$
}

\noindent\textbf{Computation of {\boldmath$i[ans]$}}: The value of $ans$ at the lowest level (level 1) is set to be the corresponding bit of $n$. At level 2, $i[ans]$ is computed to be equal to the bitstring stored in $ans$ of left child (left-shifted by $i[shift]$ bits (by 1 bit)) multiplied with $m$, added to the bitstring stored in $ans$ of right child multiplied with $m$. At every level above level 2, $i[ans]$ is set by left shifting $(2i)[ans]$ by the $i[shift]$, and then adding $(2i+1)[ans]$ to that value. Thus, to propagate the value of $ans$ among the nodes correctly, we declare them to be \imped as follows. 
{
\noindent$$\begin{array}{c}
\textsc{\Imped-Multiplication-Standard-Ans}(i)\equiv\\
    \begin{cases}
        i[ans]\neq n[i-|n|+1] & \quad \quad \text{if $i\geq|n|$.}\\
        i[ans]\neq lshift(((2i)[ans]\times m)$, $i[shift])\\ \quad \quad +((2i+1)[ans]\times m) & \quad \quad \text{if $(2i)[shift]=(2i+1)[shift]=0$.}\\
        i[ans]\neq lshift((2i)[ans]$, $i[shift])+(2i+1)[ans] & \quad \quad \text{if $(2i)[shift]=(2i+1)[shift]\geq 1$.}
    \end{cases}
\end{array}
$$
}

We observe that determining $\textsc{\Imped-Multiplication-Standard-Ans}(i)$ is more complex than determining $\textsc{\Imped-Multiplication-Standard-Shift}(i)$. However, we can eliminate it by observing that it suffices to update $ans$ when $shift$ is updated. This requires that $i[shift]$ and $i[ans]$ are updated at a node $i$ atomically in a single step. In that case, we can view the algorithm as \Cref{algorithm:mul-normal-tree}.

\begin{algorithm}Parallelized standard multiplication algorithm.\label{algorithm:mul-normal-tree}
\end{algorithm}
$
\begin{array}{|l|}
    \hline 
    \text{\textit{Rules for node $i$}}.\\~\\
    \textsc{\Imped-Multiplication-Standard-Ans}(i)\longrightarrow\\
    \begin{array}{l}
        \begin{cases} 
            i[shift]=0,i[ans] = n[i-|n|+1]. & \text{if $i\geq|n|$.}\\
            i[shift]=1,\\ \quad \quad i[ans]=lshift(((2i)[ans]\times m)$, $1)+((2i+1)[ans]\times m). & \text{if $(2i)[shift]=0$.}\\
            i[shift]=2\times (2i)[shift],\\ \quad \quad i[ans]=lshift((2i)[ans]$, $i[shift])+(2i+1)[ans]. &  \text{otherwise.}
        \end{cases}
    \end{array}~\\
    \hline
\end{array}
$

From the above description, we see that the standard multiplication algorithm satisfies the constraints of lattice-linearity, {which we prove in the following part of this subsection}. This algorithm executes in $O(|m|\times\lg |n|)$ time. Its work complexity is $O(|n|\times |m|)$, which is same as the time complexity of the standard multiplication algorithm. The example below demonstrates the working of \Cref{algorithm:mul-normal-tree}.

\begin{example}
    In \Cref{figure:11011times101}, we demonstrate the multiplication of the bitstrings 00011011 (value of $n$) and 0101 (value of $m$) under  \Cref{algorithm:mul-normal-tree}. 
    \qed 
\end{example}

\begin{figure}[ht]
    \centering
    \begin{tikzpicture}[scale=1.3,every node/.style={scale=.7}]
        \node [draw=black,label=below:node 1] (a1) at (1.5,3) {1010000+110111};
        
        \node [draw=black,label=below:node 2] (a2) at (.5,2) {0+101};
        \node [draw=black,label=below:node 3] (a3) at (4.5,2) {101000+1111};
       
        \node [draw=black,label=below:node 4] (a4) at (-.5,1) {0+0};
        \node [draw=black,label=below:node 5] (a5) at (1.5,1) {0+101};
        \node [draw=black,label=below:node 6] (a6) at (3.5,1) {1010+0};
        \node [draw=black,label=below:node 7] (a7) at (5.5,1) {1010+101};
        
        \node [draw=black,label=below:node 8] (a8) at (-1,0) {0};
        \node [draw=black,label=below:node 9] (a9) at (0,0) {0};
        \node [draw=black,label=below:node 10] (a10) at (1,0) {0};
        \node [draw=black,label=below:node 11] (a11) at (2,0) {1};
        \node [draw=black,label=below:node 12] (a12) at (3,0) {1};
        \node [draw=black,label=below:node 13] (a13) at (4,0) {0};
        \node [draw=black,label=below:node 14] (a14) at (5,0) {1};
        \node [draw=black,label=below:node 15] (a15) at (6,0) {1};
        
        \draw[->] (a1) -- (a2); \draw[->] (a1) -- (a3);
        \draw[->] (a2) -- (a4); \draw[->] (a2) -- (a5); \draw[->] (a3) -- (a6); \draw[->] (a3) -- (a7);
        \draw[->] (a4) -- (a8); \draw[->] (a4) -- (a9); \draw[->] (a5) -- (a10); \draw[->] (a5) -- (a11); \draw[->] (a6) -- (a12); \draw[->] (a6) -- (a13); \draw[->] (a7) -- (a14); \draw[->] (a7) -- (a15);
        
        \node at (-2,0) {$shift = 0$}; \node at (-2,1) {$shift = 1$}; \node at (-2,2) {$shift = 2~(10)$}; \node at (-2,3) {$shift = 4~(100)$};
        
        \node at (4,3) {(=10000111)};
    \end{tikzpicture}
    \caption{Multiplication of 00011011 and 0101 in base 2. }
    \label{figure:11011times101}
\end{figure}

\subsubsection{Lattice-Linearity}\label{subsubsection:ll-mul-normal}

\begin{lemma}\label{lemma:mul-normal-predicate-ll}
    Given the input bitstrings $n$ and $m$, the predicate 
    \begin{center}
        $\forall i:\lnot\textsc{\Imped-Multiplication-Standard-Ans}(i)$
    \end{center}
    is lattice-linear on $2|n|-1$ computing nodes.
\end{lemma}

\begin{proof}
    Let us assume that node 1 does not have the correct value of $1[ans]=n\times m$. This implies that (1) node 1 has a non-updated value in $1[ans]$ or $1[shift]$, in which case node 1 is \imped, or (2) node 2 does not have the correct values $2[ans]=n[1:|n|/2]\times m$ or $2[shift]$ or node 3 does not have the correct values $3[ans]=n[|n|/2+1:|n|]\times m$ or $3[shift]$.
    
    Recursively, this can be extended to any node $i$. Let that node $i$ has stored an incorrect value in $i[ans]$ or $i[shift]$. If $i\leq |n|-1$, then this means that (1) node $i$ has a non-updated value in $i[ans]$ or $i[shift]$, in which case node $i$ is \imped, or (2) node $2i$ or node $2i+1$ do not have the correct values in $(2i)[ans]=n[|n|-2(|n|-2i)+1:|n|-2(|n|-2i)+2]$ or $(2i+1)[ans]=n[|n|-2(|n|-2i-1)+1:|n|-2(|n|-2i-1)+2]$ respectively. If $i\geq |n|$, then this implies that node $i$ has not read the correct value $n[i-|n|+1]$, in which case, again, node $i$ is \imped.
    
    From these cases, we have that given a global state $s$, where $s=\langle\langle 1[ans]$, $1[shift]\rangle,$ $\langle 2[ans]$, $2[shift]\rangle$, $...$, $\langle (2|n|-1)[ans]$, $(2|n|-1)[shift]\rangle\rangle$, if $s$ is \imped, there is at least one node which is \imped.

    
    Next, we show that if some node is \imped, then node 1 will not store the correct answer. $\forall i:i\in[1:2|n|-1]$ node $i$ is \imped if it has a non-updated value in $i[ans]$ ($i[ans]=n[|n|-2(|n|-i)+1:|n|-2(|n|-i)+2]\times m$ if $i\leq |n|-1$ and $i[ans]=n[i-|n|+1]$ if $i\geq |n|$) or $i[shift]$. This implies that the parent of node $i$ will also store incorrect value in its $ans$ or $shift$ variable. Recursively, we have that node 1 stores an incorrect value in $1[ans]$, and thus the global state is \imped.
\end{proof}

{
\begin{lemma}\label{lemma:mul-normal-predicate-self-stabilizing}
   The predicate
    \begin{center}
        $\forall i:\lnot\textsc{\Imped-Multiplication-Standard-Ans}(i)$
    \end{center}
    is a lattice-linear self-stabilizing predicate.
\end{lemma}
\begin{proof}
Since the lattice linearity was shown in \Cref{lemma:mul-normal-predicate-ll}, we only focus on the self-stabilizing aspect here. To show this, we need to show that this predicate is true in the supremum state. 

    We note that if $s$ is the supremum of the induced lattice, then there is no outgoing edge from $s$ to any other global state in the state transition graph. It means that in $s$, no node is enabled, and so, no node is \imped. Thus, we have that the predicate 
    \begin{center}
        $\forall i:\lnot\textsc{\Imped-Multiplication-Standard-Ans}(i)$
    \end{center}
    holds true in $s$.
    Since $s$ is an arbitrary supremum, this predicate 
    is a self-stabilizing predicate.
\end{proof}
}

{
\begin{theorem}\label{theorem:algo-mul-normal-algo-ss}
    \Cref{algorithm:mul-normal-tree} is silent and self-stabilizing.
\end{theorem}
\begin{proof}
    
    The nodes that have ID $\geq n$ (leaf nodes) read the bit-values directly from the input (cf. first rule of \Cref{algorithm:mul-normal-tree}), so their value is fixed immediately when they make their first move. 
    Then, these nodes will not change their states. 
    After that, recursively, all other nodes will update their state with respect to the state of their children (cf. second and third rules of \Cref{algorithm:mul-normal-tree}). If the nodes are arbitrarily initialized, then several nodes may need to update their state more than once.
    
    This process will continue for all nodes in the tree, and the nodes will converge to a stable state bottom-to-top, in an acyclic fashion. Therefore, eventually, the root node will correct its own state. At this point, no node is enabled and the value of $ans$ in the root node provides the answer. Thus, \Cref{algorithm:mul-normal-tree} is silent and self-stabilizing.
    
    
    
    
\end{proof}
}



\subsection{Parallelized Karatsuba's Multiplication Operation}\label{subsection:karatsuba}

In this section, we study the lattice-linearity of the parallelization (of Karatsuba's \cite{Karatsuba1962} algorithm) that was presented in \cite{Cesari1996}. First we discuss the idea behind the sequential Karatsuba's algorithm, and then we elaborate on its lattice-linearity.

\subsubsection{Key Idea of the Sequential Karatsuba's Algorithm}\label{sububsection:idea-karatsuba}

The input is a pair of bitstrings $n$ and $m$.
This algorithm is recursive in nature. As the base case, when the length of $n$ and $m$ equals 1 then, the multiplication result is trivial. 
When the length is greater than 1, we let $n=append(a, b)$ and $m=append(c, d)$, where  $a$ and $b$ are half the length of $n$, and $c$ and $d$ are half the length of $m$. Here, $append(a, b)$, for example, represents concatenation of $a$ and $b$, which equals $n$.

Let $z=2^{|b|}$. 
$n\times m$ can be computed as $a\times c\times z^2+(a\times d+b\times c)\times z+b\times d$. $a\times d+b\times c$ can be computed as $(a+b)\times (c+d)-a\times c-b\times d$.
Thus, to compute $n \times m$, it suffices to compute 3 multiplications $a\times c$, $b\times d$ and $(a+b)\times (c+d)$.
Hence, we can eliminate one of the multiplications.
In the following section, we analyse the lattice-linearity of the parallelization of this algorithm as described in \cite{Cesari1996}.

\subsubsection{The CM parallelization \cite{Cesari1996} for Karatsuba's algorithm}

The Karatsuba multiplication algorithm involves dividing the input string into substrings and use them to evaluate the multiplication recursively. In the parallel version of this algorithm, the recursive call is replaced by utilizing other (\textit{children}) nodes to treat those substrings. We elaborate more on this in the following paragraphs. Consequently, 
this algorithm induces a tree among the computing nodes, where every non-leaf node has three children. 
This algorithm works in two phases, top-down and bottom-up. This algorithm uses four variables to represent the state of each node $i$: $i[n]$, $i[m]$, $i[ans]$ and $i[shift]$ respectively.

In the sequential Karatsuba's algorithm, both of the input strings $n$ and $m$ are divided into two substrings each, and the algorithm then runs recurively on three different input pairs computed from those excerpt bitstrings. In the parallel version, those recursive calls are replaced by \textit{activating} three children nodes \cite{Cesari1996}. As a result of such parallelization, if there is no carry-forwarding due to addition, we require $\lg |n|$ levels, for which a total of $|n|^{\lg 3}$ nodes are required. However, if there is carry-forwarding due to additions, then we require $2 \lg |n|$ levels, for which a total of $|n|^{2\lg 3}$ nodes are required.

In the top-down phase, if $|i[m]| > 1$ or $|i[n]|>1$, then $i$ writes (1) $a$ and $c$ to its left child, node $3i-1$ ($(3i-1)[m]=a$ and $(3i-1)[n]=c$), (2) $b$ and $d$ to its middle child, node $3i$ ($(3i)[m]=b$ and $(3i)[n]=d$), and (3) $a+b$ and $c+d$ to its right child, node $3i+1$ ($(3i+1)[m]=a+b$ and $(3i+1)[n]=c+d$). 
If $|i[m]|=|i[n]|=1$, i.e., in the base case, the bottom-up phase begins and node $i$ sets $i[ans]=i[m]\times i[n]$ that can be computed trivially since $|i[m]|=|i[n]|=1$.

In the bottom-up phase, node $i$ sets 
$i[ans] = (3i-1)[ans]\times z^2 + ((3i+1)[ans]-((3i-1)[ans]+(3i)[and]))\times z + (3i)[ans]$.
Notice that multiplication by $z$ and $z^2$ corresponds to bit shifts and does not need an actual multiplication. 
Consequently, the product of $m\times n$ for node $i$ is computed by this algorithm.

With some book-keeping (storing the place values of most significant bits of $a+b$ and $c+d$), a node $i$ only needs to write the rightmost $\frac{|i[m]|}{2}$ and $\frac{|i[n]|}{2}$ bits to its children. Thus, we can safely assume that when a node writes $m$ and $n$ to any of its children, then $m$ and $n$ of that child are of equal length and are of length in some power of 2. (If we do not do the book-keeping, the required number of nodes increases, this number is upper bounded by $|n|^{2\lg 3}$ as the number of levels is upper bounded by $2\lg |n|$; this observation was not made in \cite{Cesari1996}.) However, we do not show such book-keeping in the algorithm for brevity. Thus this algorithm would require $2\lg |n|$ levels, i.e., $|n|^{2\lg 3}$ nodes.\\~

\noindent\textbf{\textit{Computation of \textit{i[shift]}}}:
This algorithm utilizes $shift$ to compute $z$. A node $i$ updates $i[shift]$ by doubling the value of $shift$ from its children. A node $i$ evaluates that it is \imped because of an incorrect value of $i[shift]$ by evaluating the following macro. 
{
$$
\begin{array}{c}
    \textsc{\Imped-Multiplication-Karatsuba-Shift}(i)\equiv\\
    \begin{array}{l}
        \begin{cases}
            |i[m]|=1\land |i[n]|=1\land i[shift]\neq 0~~OR\\
            (3i)[shift]=(3i-1)[shift]=0\leq (3i+1)[shift]\land i[shift]\neq 1~~OR\\
            0<(3i)[shift]=(3i-1)[shift]\leq (3i+1)[shift]\land i[shift]\neq (3i)[shift]\times 2.\\
        \end{cases}
    \end{array}
\end{array}
$$
}

\noindent\textbf{\textit{Computation of \textit{i[m]} and \textit{i[n]}}}:
To ensure that the data flows down correctly, we declare a node $i$ to be \imped as follows.
{
$$\begin{array}{c}
    \textsc{\Imped-Multiplication-Karatsuba-TopDown}(i)\equiv\\
    \begin{array}{l}
        \begin{cases}
            i=1\land (i[m]\neq m\lor i[n]\neq n)~~OR\\
            ((|i[m]|>1\land |i[n]|>1)\land\\
            ((3i-1)[m] \neq i[m]\Big[1:\frac{|i[m]|}{2}\Big]~~OR\\
            (3i-1)[n] \neq i[n]\Big[1:\frac{|i[n]|}{2}\Big]~~OR\\
            (3i)[m] \neq i[m]\Big[\frac{|i[m]|}{2}+1:|i[m]|\Big]~~OR\\
            (3i)[n] \neq i[n]\Big[\frac{|i[n]|}{2}+1:|i[n]|\Big]~~OR\\
            (3i+1)[m]\neq i[m]\Big[1:\frac{|i[m]|}{2}\Big]+i[m]\Big[\frac{|i[m]|}{2}+1:|i[m]|\Big]~~OR\\
            (3i+1)[n]\neq i[n]\Big[1:\frac{|i[n]|}{2}\Big]+i[n]\Big[\frac{|i[n]|}{2}+1:|i[n]|\Big])).\\
        \end{cases}
    \end{array}
\end{array}
$$
}

\noindent\textbf{\textit{Computation of \textit{i[ans]}}}:
To determine if a node $i$ has stored $i[ans]$ incorrectly, it evaluates to be \imped as follows.
{
$$\begin{array}{c}
    \textsc{\Imped-Multiplication-Karatsuba-BottomUp}(i)\equiv\\
    \begin{array}{l}
        \begin{cases}
            |i[m]|=1 \land |i[n]|=1\land i[ans]\neq i[m]\times i[n]~~OR\\
            |i[m]|>1\land |i[n]|>1\land (i[ans]\neq lshift((3i-1)[ans], i[shift])\\
            \quad \quad +lshift((3i+1)[ans]-(3i-1)[ans]-(3i)[ans],(3i)[shift])\\
            \quad \quad +(3i+1)[ans])
        \end{cases}
    \end{array}
\end{array}
$$
}

Thus, \Cref{algorithm:parallel-karatsuba}
 is described as follows:

\begin{algorithm}\label{algorithm:parallel-karatsuba}Parallel processing version of Karatsuba's algorithm.
\end{algorithm}
{
$
\begin{array}{|l|}
    \hline
    \text{\textit{Rules for node $i$}}.\\~\\
    \textsc{\Imped-Multiplication-Karatsuba-Shift}(i)\longrightarrow\\
    \begin{array}{l}
        \begin{cases}
            i[shift]=0 & \text{if $|i[m]|=1\land |i[n]|=1\land i[shift]\neq 0$.}\\
            i[shift]=1 & \text{if $(3i)[shift]=(3i-1)[shift]=0$}\\
             & \quad \quad \text{$\leq (3i+1)[shift]\land i[shift]\neq 1$}\\
            i[shift]=(3i)[shift]\times 2 & \text{otherwise}
        \end{cases}
    \end{array}
    ~\\
    \textsc{\Imped-Multiplication-Karatsuba-TopDown}(i)\longrightarrow\\
    \begin{array}{l}
        \begin{cases}
            i[m]=m,i[n]=n & \text{if $i=1\land (i[m]\neq m\lor i[n]\neq n)$.}\\
            (3i-1)[m] = i[m]\Big[1:\frac{|i[m]|}{2}\Big] & \text{if $(3i-1)[m] \neq i[m]\Big[1:\frac{|i[m]|}{2}\Big]$.}\\
            (3i-1)[n] = i[n]\Big[1:\frac{|i[n]|}{2}\Big] & \text{if $(3i-1)[n] \neq i[n]\Big[1:\frac{|i[n]|}{2}\Big]$.}\\
            (3i)[m] = i[m]\Big[\frac{|i[m]|}{2}+1:|i[m]|\Big] & \text{if $(3i)[m] \neq i[m]\Big[\frac{|i[m]|}{2}+1:|i[m]|\Big]$.}\\
            (3i)[n] = i[n]\Big[\frac{|i[n]|}{2}+1:|i[n]|\Big] & \text{if $(3i)[n] \neq i[n]\Big[\frac{|i[n]|}{2}+1:|i[n]|\Big]$.}\\
            (3i+1)[m]=i[m]\Big[1:\frac{|i[m]|}{2}\Big]\\ \quad +i[m]\Big[\frac{|i[m]|}{2}+1:|i[m]|\Big] & \text{if $(3i+1)[m]\neq i[m]\Big[1:\frac{|i[m]|}{2}\Big]$.}\\
             & \text{\quad \quad $+i[m]\Big[\frac{|i[m]|}{2}+1:|i[m]|\Big]$.}\\
            (3i+1)[n]=i[n]\Big[1:\frac{|i[n]|}{2}\Big]\\ \quad +i[n]\Big[\frac{|i[n]|}{2}+1:|i[n]|\Big] & \text{otherwise}\\
        \end{cases}
    \end{array}
    ~\\
    \textsc{\Imped-Multiplication-Karatsuba-BottomUp}(i)\longrightarrow\\
    \begin{array}{l}
        \begin{cases}
            i[ans]=i[m]\times i[n] & \text{if $|i[m]|=1 \land |i[n]|=1$}\\
            i[ans]=lshift((3i-1)[ans], i[shift])+\\
            \quad \quad lshift((3i+1)[ans]-(3i-1)[ans]\\
            \quad \quad \quad -(3i)[ans],(3i)[shift])+(3i+1)[ans]) & \text{otherwise.}
        \end{cases}
    \end{array}
    ~\\
    \hline
\end{array}
$
}

\Cref{algorithm:parallel-karatsuba} converges in $O(|n|)$ time \cite{Cesari1996}, and its work complexity is $O(n^{\lg 3})$, which is the time complexity of the sequential Karatsuba's algorithm \cite{Cesari1996}. 

\begin{example}\label{example:1111times1111}
    \Cref{figure:100times100} evaluates $100\times 100$ following \Cref{algorithm:parallel-karatsuba}.    
\end{example}
\begin{figure}[ht]
    \centering
    \begin{minipage}{1\textwidth}
        \begin{tikzpicture}[scale=.7,every node/.style={scale=0.65}]
            \node [draw=black,label=below:node 1] (a1) at (0,0) {\begin{tabular}{c}$m=0100$\\$n=0100$\end{tabular}};
            
            \node [draw=black,,label=below:node 2] (a2) at (-6,-2.5) {\begin{tabular}{c}$m=01$\\$n=01$\end{tabular}};
            \node [draw=black,,label=below:node 3] (a3) at (0,-2.5) {\begin{tabular}{c}$m=00$\\$n=00$\end{tabular}};
            \node [draw=black,,label=below:node 4] (a4) at (6,-2.5) {\begin{tabular}{c}$m=01$\\$n=01$\end{tabular}};
            
            \node [draw=black,,label=below:node 5] (a5) at (-8,-5) {\begin{tabular}{c}$m=0$\\$n=0$\end{tabular}};
            \node [draw=black,,label=below:node 6] (a6) at (-6,-5) {\begin{tabular}{c}$m=1$\\$n=1$\end{tabular}};
            \node [draw=black,,label=below:node 7] (a7) at (-4,-5) {\begin{tabular}{c}$m=1$\\$n=1$\end{tabular}};
            \node [draw=black,,label=below:node 8] (a8) at (-2,-5) {\begin{tabular}{c}$m=0$\\$n=0$\end{tabular}};
            \node [draw=black,,label=below:node 9] (a9) at (0,-5) {\begin{tabular}{c}$m=0$\\$n=0$\end{tabular}};
            \node [draw=black,,label=below:node 10] (a10) at (2,-5) {\begin{tabular}{c}$m=0$\\$n=0$\end{tabular}};
            \node [draw=black,,label=below:node 11] (a11) at (4,-5) {\begin{tabular}{c}$m=0$\\$n=0$\end{tabular}};
            \node [draw=black,,label=below:node 12] (a12) at (6,-5) {\begin{tabular}{c}$m=1$\\$n=1$\end{tabular}};
            \node [draw=black,,label=below:node 13] (a13) at (8,-5) {\begin{tabular}{c}$m=1$\\$n=1$\end{tabular}};
            
            \node at (10.5,.5) {$shift=2$};
            \node at (10.5,-2.5) {$shift=1$};
            \node at (10.5,-5.5) {$shift=0$};
            
            \draw (a2) -- (-6,-1.5) -- (0,-1.5); \draw[->] (a3) -- (a1); \draw (a4) -- (6,-1.5) -- (0,-1.5);
            
            \draw (a5) -- (-8,-4) -- (-6,-4); \draw[->] (a6) -- (a2); \draw (a7) -- (-4,-4) -- (-6,-4);
            \draw (a8) -- (-2,-4) -- (0,-4); \draw[->] (a9) -- (a3); \draw (a10) -- (2,-4) -- (0,-4);
            \draw (a11) -- (4,-4) -- (6,-4); \draw[->] (a12) -- (a4); \draw (a13) -- (8,-4) -- (6,-4);
        \end{tikzpicture}
    \end{minipage}
    \textbf{(a)}
    \begin{minipage}{1\textwidth}
        \begin{tikzpicture}[scale=.7,every node/.style={scale=.65}]
            \node [draw=black,label=below:node 1] (a1) at (0,.5) {\begin{tabular}{c}$m=0100$\\$n=0100$\\$ans=10000$\end{tabular}};
            
            \node [draw=black,,label=below:node 2] (a2) at (-6.5,-2.5) {\begin{tabular}{c}$m=01$\\$n=01$\\$ans=1$\end{tabular}};
            \node [draw=black,,label=below:node 3] (a3) at (0,-2.5) {\begin{tabular}{c}$m=00$\\$n=00$\\$ans=0$\end{tabular}};
            \node [draw=black,,label=below:node 4] (a4) at (6.5,-2.5) {\begin{tabular}{c}$m=01$\\$n=01$\\$ans=1$\end{tabular}};
            
            \node [draw=black,,label=below:node 5] (a5) at (-8.5,-5.5) {\begin{tabular}{c}$m=0$\\$n=0$\\$ans=0$\end{tabular}};
            \node [draw=black,,label=below:node 6] (a6) at (-6.5,-5.5) {\begin{tabular}{c}$m=1$\\$n=1$\\$ans=1$\end{tabular}};
            \node [draw=black,,label=below:node 7] (a7) at (-4.5,-5.5) {\begin{tabular}{c}$m=1$\\$n=1$\\$ans=1$\end{tabular}};
            \node [draw=black,,label=below:node 8] (a8) at (-2,-5.5) {\begin{tabular}{c}$m=0$\\$n=0$\\$ans=0$\end{tabular}};
            \node [draw=black,,label=below:node 9] (a9) at (0,-5.5) {\begin{tabular}{c}$m=0$\\$n=0$\\$ans=0$\end{tabular}};
            \node [draw=black,,label=below:node 10] (a10) at (2,-5.5) {\begin{tabular}{c}$m=0$\\$n=0$\\$ans=0$\end{tabular}};
            \node [draw=black,,label=below:node 11] (a11) at (4.5,-5.5) {\begin{tabular}{c}$m=0$\\$n=0$\\$ans=0$\end{tabular}};
            \node [draw=black,,label=below:node 12] (a12) at (6.5,-5.5) {\begin{tabular}{c}$m=1$\\$n=1$\\$ans=1$\end{tabular}};
            \node [draw=black,,label=below:node 13] (a13) at (8.5,-5.5) {\begin{tabular}{c}$m=1$\\$n=1$\\$ans=1$\end{tabular}};
            
            \node at (10.5,.5) {$shift=2$};
            \node at (10.5,-2.5) {$shift=1$};
            \node at (10.5,-5.5) {$shift=0$};
            
            \draw[<-] (a2) -- (-6.5,-1) -- (0,-1); \draw[<-] (a3) -- (a1); \draw[<-] (a4) -- (6.5,-1) -- (0,-1);
            
            \draw[<-] (a5) -- (-8.5,-4) -- (-6.5,-4); \draw[<-] (a6) -- (a2); \draw[<-] (a7) -- (-4.5,-4) -- (-6.5,-4);
            \draw[<-] (a8) -- (-2,-4) -- (0,-4); \draw[<-] (a9) -- (a3); \draw[<-] (a10) -- (2,-4) -- (0,-4);
            \draw[<-] (a11) -- (4.5,-4) -- (6.5,-4); \draw[<-] (a12) -- (a4); \draw[<-] (a13) -- (8.5,-4) -- (6.5,-4);
        \end{tikzpicture}
    \end{minipage}
    \textbf{(b)}
    \caption{Demonstration of multiplication of 100 and 100 in base 2: (a) top down (b) bottom up.
    }
    \label{figure:100times100}
\end{figure}

\subsubsection{Lattice-linearity}\label{subsubsection:ll-karatsuba}

\begin{lemma}\label{lemma:mul-karatsuba}
    Given the input bitstrings $n$ and $m$, the predicate 
    \begin{center}
        $\forall i:\lnot(\textsc{\Imped-Multiplication-Karatsuba-Shift}(i)\lor$\\
        $\textsc{\Imped-Multiplication-Karatsuba-TopDown}(i) \lor$\\
        $\textsc{\Imped-Multiplication-Karatsuba-BottomUp}(i))$
    \end{center}
    is lattice-linear on $|n|^{2\lg 3}$ computing nodes.
\end{lemma}

\begin{proof}
    For the global state to be optimal, in this problem, we require node 1 to store the correct multiplication result in $1[ans]$. To achieve this, each node $i$ must have the correct value stored in $i[n]$ and $i[m]$, and their children must store correct values of $n$, $m$ and $ans$ according to the values of $i[n]$ and $i[m]$. This in turn requires all nodes to store the correct $i[shift]$ values.

    Let us assume 
    for contradiction
    that node 1 does not store the correct value in $1[ans]$ as $n\times m$. This implies that (1) node 1 does not have an updated value in $1[n]$ or $1[m]$, or (2) node 1 has a non-updated value of $1[ans]$, (3) node 1 has not written the updated values to $2[n]$ \& $2[m]$, $3[n]$ \& $3[m]$ or $4[n]$ \& $4[m]$, (4) node 1 has a non-updated value in $1[shift]$, or (5) nodes 2, 3 or 4 have incorrect values in their respective $n$, $m$, $ans$ or $shift$ variables. In cases (1),...,(4), node 1 is \imped.
    
    Recursively, this can be extended to any node $i$. Let node $i$ has stored an incorrect value in $i[ans]$ or $i[shift]$. Let $i>1$. Then (1) node $i$ has a non-updated value in $i[shift]$, $i[ans]$, $i[n]$ or $i[m]$, or (2) if $|i[m]|>1$ or $|i[n]|>1$, node $i$ has not written updated values to $(3i-1)[n]$ \& $(3i-1)[m]$ or $(3i)[n]$ \& $(3i)[m]$ or $(3i+1)[n]$ \& $(3i+1)[m]$, in which case node $i$ is \imped. In both these cases, node $i$ is \imped. It is also possible that at least one of the children of node $i$ has incorrect values in its respective $n$, $m$, $ans$ or $shift$ variables.
    
    From these cases, we have that given a global state $s$, where $s=\langle\langle 1[n],$ $1[m]$, $1[ans]\rangle$, $\langle 2[n]$, $2[m]$, $2[ans]\rangle$, $...$
    $\rangle$, if $s$ is \imped, there is at least one node which is \imped. This shows that if the global state is \imped, then there exists some node $i$ which is \imped.
    
    Next, we show that if some node is \imped, then node 1 will not store the correct answer. Node 1 is \imped if it has not read the correct value $1[m]$ and $1[n]$. Additionally, $\forall i:i\in[1:n^{2\lg 3}]$ node $i$ is \imped if (1) it has non-updated values in $i[ans]$ or $i[shift]$, (2) $i$ has not written the correct values to $(3i-1)[n]$ \& $(3i-1)[m]$ or $(3i)[n]$ \& $(3i)[m]$ or $(3i+1)[n]$ \& $(3i+1)[m]$. This implies that the parent of $i$ will also store incorrect value in its $ans$ or $shift$ variable. Recursively, we have that node 1 stores an incorrect value in $1[ans]$. Thus, the global state is \imped.
    %
\end{proof}

{
With the arguments similar to those made in the proof of \Cref{lemma:mul-normal-predicate-self-stabilizing} and \Cref{theorem:algo-mul-normal-algo-ss}, we have the following.
\begin{lemma}
    The predicate
    \begin{center}
        $\forall i:\lnot(\textsc{\Imped-Multiplication-Karatsuba-Shift}(i)\lor$\\
        $\textsc{\Imped-Multiplication-Karatsuba-TopDown}(i) \lor$\\
        $\textsc{\Imped-Multiplication-Karatsuba-BottomUp}(i))$
    \end{center}
    is a lattice-linear self-stabilizing predicate.
\end{lemma}

\begin{theorem}
    \Cref{algorithm:parallel-karatsuba} is silent and self-stabilizing.
\end{theorem}
}

\noindent\textbf{\textit{Remark}}: 
{In \Cref{algorithm:parallel-karatsuba}, an \imped node updates the state of its children. We have done so for the brevity of the presentation of the algorithm. In practice, each child will notice that its local state does not tally with the state of its parent, and will then update its own state. 
}


\section{Parallel processing modulo operation}\label{section:parallel-modulo}\label{section:mod-parallel}

In this section, we demonstrate parallel processing systems which can be used to compute a given modulo operation $n \mod m$.

\subsection{Using $|$\textit{n}$|$ processors}\label{subsection:mod-n-processors}

In this subsection, we discuss a lattice-linear method to compute modulo operation which requires $|n|$ processors.
First, we discuss the key idea of the sequential algorithm, then we elaborate on the lattice-linearity of its parallelization.

\subsubsection{Key Idea}

This algorithm is based on the standard modulo operation. Suppose that we have computed $a=n[1:|n|-1]\mod m$ and $b=n[|n|]$. Then we have that the resultant value of $n\mod m$ is $(a\times 2+b)\mod m$, which is also equal to $(lshift(a,1)+b)\mod m$.

\subsubsection{Parallelization}

Every node, sequentially, reads a distinct bit of the input dividend $n$. Every node $i$ will eventually store the value of $n[1:i]$ under modulo $m$. The last node, indexed as node $|n|$, will store the final value, i.e. $n\mod m$. We demonstrate two ways of executing this algorithm. One way is with using the machine $M$ that we constructed in \Cref{subsection:dfa-modulo}. Another way is to perform the computation without $M$ where we use $\textsc{Division-Modulo}()$ that we defined in \Cref{subsection:long-division}. We demonstrate these methods in the following.

\paragraph{Using \textit{M}}~\\
In this part, we will utilize $M$ to compute $n\mod m$. Since every node $i$ must store the value of $n[1:i]\mod m$, the \imped node $i$ can be defined as follows.
{
\begin{equation*}
    \textsc{\Imped-Linear-Modulo}(i)\equiv
    \begin{cases}
        i[ans] \neq n[1] & \text{if $i=1$}\\
        (i[ans]\neq M((i-1)[ans],n[i]) & \text{otherwise}
    \end{cases}
\end{equation*}
}
In the definition of $\textsc{\Imped-Linear-Modulo}(i)$, $M(i,j)$ means that $M$ is being invoked with an initial state $q_i$ and an input $j\in\{0,1\}$, i.e. $M(i,j)=\delta(q_i,j)$. If $\delta(q_i,j)=q_k$, then the execution of $M(i,j)$ will give $k$ as output. The algorithm to compute $n\mod m$ is demonstrated in \Cref{algorithm:linear-modulo-with-M}.

\begin{algorithm}\label{algorithm:linear-modulo-with-M} Computing modulo on $|n|$ processors using $M$.
\end{algorithm}
    $$
    \begin{array}{|l|}
        \hline
        \text{Rules for node $i$}.\\~\\
        \begin{array}{c}
            \textsc{\Imped-Linear-Modulo}(i)\longrightarrow\\
            \begin{cases}
                i[ans] = n[1] & \text{if $i=1$}.\\
                i[ans]=M((i-1)[ans],n[i]) & \text{otherwise}.
            \end{cases}
        \end{array}
        ~\\
        \hline
    \end{array}
    $$

The time complexity of this algorithm is $O(|n|)$. 
However, this method needs a preprocessing of $O(m\times |m|)$, which is quite high and impractical, especially if $m$ is large. We present this result only to demonstrate that some pre-processing can reduce the complexity of the modulo operation substantially.

\begin{theorem}
    Given the input bitstrings $n$ and $m$, the predicate 
    \begin{center}
        $\forall i:\lnot\textsc{\Imped-Linear-Modulo}(i)$
    \end{center}
    is lattice-linear on $|n|$ computing nodes.
\end{theorem}

\begin{proof}
    Let us assume that node $|n|$ has incorrect value in $|n|[ans]$. This implies that (1) node $|n|$ does not have an updated value in $(|n|)[ans]$, in which case node $|n|$ is \imped, or (2) node $|n|-1$ has an incorrect value in $(|n|-1)[ans]$.
    
    Recursively, this can be extended to any node $i$. Let that node $i$ has stored an incorrect value in $i[ans]$, then (1) node $i$ has a non-updated value in $i[ans]$, in which case, node $i$ is imedensable, or (2) node $i-1$ has an incorrect value in $(i-1)[ans]$. 
    From these cases, we have that given a global state $s$, where $s=\langle 1[ans], 2[ans],..., |n|[ans]\rangle$, if $s$ is \imped, there is at least one node which is \imped.
    
    This shows that if the global state is \imped, then there exists some node $i$ which is \imped.
    
    Next, we show that if some node is \imped, then node 1 will not store the correct answer. If node $i$ is \imped, then node $i$ has a non-updated value in $i[ans]$. 
    This implies that node $i+1$ will also store incorrect value in $(i+1)[ans]$. Recursively, we have that node $|n|$ stores an incorrect value in $|n|[ans]$, and thus the global state is \imped.
    %
\end{proof}

{
With the arguments similar to those made in the proof of \Cref{lemma:mul-normal-predicate-self-stabilizing} and \Cref{theorem:algo-mul-normal-algo-ss}, we have the following.
\begin{lemma}
    The predicate
    \begin{center}
        $\forall i:\lnot\textsc{\Imped-Linear-Modulo}(i)$
    \end{center}
    is a lattice-linear self-stabilizing predicate.
\end{lemma}

\begin{theorem}
    \Cref{algorithm:linear-modulo-with-M} is silent and self-stabilizing.
\end{theorem}
}

\paragraph{Using long division}~\\
If we utilize $\textsc{Division-Modulo}()$ instead of $M$ in \Cref{algorithm:linear-modulo-with-M} (and, subtraction in place of $M$ in the definition of $\textsc{\Imped-Linear-Modulo}(i)$), then
every node takes $O(|m|)$ time because of the subtraction in \textsc{Division-Modulo}(), which implies that the total work complexity is $O(|n|\times|m|)$. Node $i$ will compute the value of $n[1:i]\mod m$ by the end of time-step $t$.
Therefore, the time complexity of this algorithm is $O(|n|\times|m|)$ to compute $n\mod m$, which is the same as the work complexity of this algorithm.

\subsubsection{Discussion}

The behaviour of these methods is lattice-linear, but similar to a uniprocessor computation, in the sense that if we ran these methods on a uniprocessor machine, then it will take the same order of time. In \Cref{subsection:mod->n-processors}, we present algorithms which exploit the power of a distributed system better. 

\subsection{Using 4 $|$\textit{n}$|$/$|$\textit{m}$|$ $-$ 1 processors}\label{subsection:mod->n-processors}

In this section, we present a parallel processing algorithm to compute $n\mod m$ using $4\times |n|/|m|-1$ computing nodes. 
First, we discuss the key idea of the sequential algorithm, then we elaborate on the lattice-linearity of its parallelization.

\subsubsection{Key Idea}

This algorithm better parallizes the idea discussed in \Cref{subsection:mod-n-processors}. Suppose that we have computed $a=n[1:\lfloor |n|/2\rfloor]\mod m$ and $b=n[\lfloor |n|/2\rfloor+1:|n|]$. Then the resultant value of $n\mod m$ is $(a\times 2^{\lceil |n|/2\rceil}+b)\mod m$, which is also equal to $(lshift(a,\lceil |n|/2\rceil)+b)\mod m$.

\subsubsection{Parallelization}

The algorithm induces a binary tree among the nodes based on their ids; there are $2\times|n|/|m|$ nodes in the lowest level (level 1).
This algorithm starts from the leaves where all leaves compute and store, in sequence, a substring of $n$ of length $|m|/2$ under modulo $m$. In the induced binary tree, the computed modulo result by sibling nodes at level $\ell$ is sent to the parent. Consecutively, those parents at level $\ell+1$, contiguously, store a larger substring of $n$ (double the bits that each of their children covers) under modulo $m$. We elaborate this procedure in this subsection.
This algorithm uses three variables to represent the state of each node $i$: $i[shift]$, $i[pow]$ and $i[ans]$.\\~

\noindent\textbf{Computation of \textit{i[shift]}: } 
The variable $shift$ stores the required power of 2. At any node at level 1, $shift$ is $0$. At level 2, the value of $shift$ at any node is $|m|/2$. At any higher level, the value of $shilft$ is twice the value of shift of its children. 
\textsc{\Imped-Log-Modulo-Shift}, in this context, is defined below. 

{
$$\begin{array}{c}
\textsc{\Imped-Log-Modulo-Shift}(i)\equiv\\
    \begin{cases}
        i[shift]\neq 0 & \text{if $i\geq 2\times|n|/|m|$}\\
        i[shift]\neq |m|/2 & \text{if $(2i)[shift]=(2i+1)[shift]=0$}\\
        i[shift]\neq2\times (2i)[shift] & \text{if $(2i)[shift]=(2i+1)[shift]\geq |m|/2$}
    \end{cases}
\end{array}
$$
}

\noindent\textbf{Computation of \textit{i[pow]}: } 
The goal of this computation is to set $i[pow]$ to be  $2^{i[shift]} \mod m$, whenever the level of $i$ is greater than 1. This can be implemented using the following definition for \textsc{\Imped-Log-Modulo-Pow}. 

$$
\begin{array}{l}
\textsc{\Imped-Log-Modulo-Pow}(i)\equiv\\
    \begin{cases}
        i[pow]\neq 1 & \text{if $i[shift]=0$}\\
        i[pow]\neq 2^{\frac{|m|}{2}} & \text{if $i[shift]=|m|/2$}\\
        i[pow]\neq((2i)[pow])^2\mod m & \text{otherwise}
    \end{cases}
\end{array}
$$

By definition, $i[pow]$ is less than $m$. Also, computation of $pow$ requires multiplication of two numbers that are upper-bounded by $|m|$. Hence, this computation can benefit from parallelization of \Cref{algorithm:parallel-karatsuba}. However, as we will see later, the complexity of this algorithm (for modulo) is dominated by the modulo operation happening in individual nodes which is $O(|m|^2)$, we can use the sequential version of Karatsuba's algorithm for multiplication, without affecting the order of the time complexity of this algorithm. 

\noindent\textbf{Computation of \textit{i[ans]}: }
We split $n$ into strings of size $\frac{|m|}{2}$, the number representing this substring is less than $m$. At the lowest level (level 1), $i[ans]$ is set to the corresponding substring. 
At higher levels, $i[ans]$ is set to $(i[pow]\times (2i)[ans] + (2i+1)[ans]) \mod m$. This computation also involves multiplication of two numbers whose size is upper bounded by $|m|$. An \imped node $i$ from a non-updated $i[ans]$ can be evaluated using \textsc{\Imped-Log-Modulo-Ans}$(i)$.

{
$$\begin{array}{c}
\textsc{\Imped-Log-Modulo-Ans}(i)\equiv\\
    \begin{cases}
        i[ans]\neq n[(i-2\times\dfrac{|n|}{|m|})\times \dfrac{|m|}{2}+1:(i-2\times\dfrac{|n|}{|m|}+1)\times \dfrac{|m|}{2}] & \text{if $i[shift]=0$}\\
        i[ans]\neq\textsc{Mod}(\textsc{Sum}(\textsc{Mul}((2i)[ans],i[pow]),(2i+1)[ans]),m) & \text{otherwise}
    \end{cases}
\end{array}
$$
}

We describe the algorithm as \Cref{algorithm:log-modulo}.

\begin{algorithm}\label{algorithm:log-modulo} Modulo computation by inducing a tree among the nodes.
\end{algorithm}
$\begin{array}{|l|}
        \hline 
        \text{\textit{Rules for node} $i$.}\\~\\
        \textsc{\Imped-Log-Modulo-Shift}(i)\longrightarrow\\
        \begin{array}{ll}
            \begin{cases}
                    i[shift]= 0 & \text{if $i\geq 2\times|n|/|m|$}\\
                i[shift]= \dfrac{|m|}{2} & \text{if $(2i)[shift]=(2i+1)[shift]=0$}\\
                i[shift]= 2\times (2i)[shift] & \text{if $(2i)[shift]=(2i+1)[shift]\geq |m|/2$}
            \end{cases}
        \end{array}
        ~\\
        \textsc{\Imped-Log-Modulo-Pow}(i)\longrightarrow\\
        \begin{array}{ll}
            \begin{cases}
                i[pow]= 1 & \text{if $i[shift]=0$}\\
                i[pow]= 2^{\frac{|m|}{2}} & \text{if $i[shift]=|m|/2$}\\
                i[pow]= \textsc{Mod}(\textsc{Mul}(i[pow],i[pow]),m) & \text{otherwise}
            \end{cases}
        \end{array}
        ~\\
        \textsc{\Imped-Log-Modulo-Ans}(i)\longrightarrow\\
        \begin{array}{ll}
            \begin{cases}
                i[ans]= n[(i-2\times\dfrac{|n|}{|m|})\times \dfrac{|m|}{2}+1:(i-2\times\dfrac{|n|}{|m|}+1)\times \dfrac{|m|}{2}] & \text{if $i[shift]=0$}\\
                i[ans]= \textsc{Mod}(\textsc{Sum}(\textsc{Mul}((2i)[ans],i[pow]),(2i+1)[ans]),m) & \text{otherwise}
            \end{cases}
        \end{array}
        ~\\
        \hline 
    \end{array}
$

\begin{example}
    \Cref{figure:11011mod11} shows the computation of $11011\mod 11$ as performed by \Cref{algorithm:log-modulo}.
    
\end{example}
\begin{figure}[ht]
    \centering
    \begin{tikzpicture}[scale=1.1,every node/.style={scale=.65}]
        \node [draw=black,label=below:{node 1}] (a1) at (3.5,3) {1+10 mod 11};
        
        \node [draw=black,label=below:{node 2}] (a2) at (1.5,2) {1 mod 11};
        \node [draw=black,label=below:{node 3}] (a3) at (5.5,2) {10+0 mod 11};
        
        \node [draw=black,label=below:{node 4}] (a4) at (.5,1) {00 mod 11};
        \node [draw=black,label=below:{node 5}] (a5) at (2.5,1) {01 mod 11};
        \node [draw=black,label=below:{node 6}] (a6) at (4.5,1) {10 mod 11};
        \node [draw=black,label=below:{node 7}] (a7) at (6.5,1) {11 mod 11};
        
        \node [draw=black,label=below:{node 8}] (a8) at (0,0) {0};
        \node [draw=black,label=below:{node 9}] (a9) at (1,0) {0};
        \node [draw=black,label=below:{node 10}] (a10) at (2,0) {0};
        \node [draw=black,label=below:{node 11}] (a11) at (3,0) {1};
        \node [draw=black,label=below:{node 12}] (a12) at (4,0) {1};
        \node [draw=black,label=below:{node 13}] (a13) at (5,0) {0};
        \node [draw=black,label=below:{node 14}] (a14) at (6,0) {1};
        \node [draw=black,label=below:{node 15}] (a15) at (7,0) {1};
        
        \draw[->] (a1) -- (a2); \draw[->] (a1) -- (a3);
        \draw[->] (a2) -- (a4); \draw[->] (a2) -- (a5); \draw[->] (a3) -- (a6); \draw[->] (a3) -- (a7);
        \draw[->] (a4) -- (a8); \draw[->] (a4) -- (a9); \draw[->] (a5) -- (a10); \draw[->] (a5) -- (a11); \draw[->] (a6) -- (a12); \draw[->] (a6) -- (a13); \draw[->] (a7) -- (a14); \draw[->] (a7) -- (a15);
        
        \node at (-2.5,0) {$pow = 1$};
        \node at (-2.5,1) {$pow = 10$};
        \node at (-2.5,2) {$pow = 100\mod 11~(=1)$}; \node at (-2.5,3) {$pow = 10000\mod 11~(=1)$};
        
        \node at (6,3) {(=0)};
    \end{tikzpicture}
    \caption{Processing $11011 \mod 11$ following \Cref{algorithm:log-modulo}.}
    \label{figure:11011mod11}
\end{figure}

\subsubsection{Lattice-Linearity}

\begin{theorem}\label{theorem:mod-log}
    Given the input bitstrings $n$ and $m$, the predicate 
    \begin{center}
        $\forall i:\lnot(\textsc{\Imped-Log-Modulo-Shift}(i)\lor$\\
        $\textsc{\Imped-Log-Modulo-Pow}(i)\lor$\\
        $\textsc{\Imped-Log-Modulo-Ans}(i))$
    \end{center}
    is lattice-linear on $4|n|/|m|-1$ computing nodes.
\end{theorem}

\begin{proof}
    For the global state to be optimal, under this algorithm, we require node 1 to store the correct modulo result in $1[ans]$. To achieve this, each node $i$ must have the correct value of $i[ans]$. This in turn requires all nodes to store the correct $i[shift]$ and $i[pow]$ values.

    Let us assume
    for contradiction
    that node 1 has incorrect value in $1[ans]$. This implies that (1) node 1 has a non-updated value in $i[shift]$ or $i[pow]$, or (2) node 1 does not have an updated value in $i[ans]$. In both these cases, node 1 is \imped. It is also possible that node 2 or node 3 have an incorrect value in their variables.
    
    Recursively, this can be extended to any node $i$. Let node $i$ has stored an incorrect value in $i[ans]$. If $i<2(|n|/|m|)$, then (1) node $i$ has a non-updated value in $i[ans]$, $i[pow]$ or $i[shift]$, in which case node $i$ is \imped or (2) node $2i$ or node $2i+1$ have incorrect value in their respective $shift$, $pow$ or $ans$ variables. 
    If $i\geq 2(|n|/|m|)$, then $i$ does not have values $i[shift]=0$, $i[pow]=1$ or a correct $i[ans]$ value, in which case $i$ is \imped.
    From these cases, we have that given a global state $s$, where $s=\langle\langle 1[shift]$, $1[pow]$, $1[ans]\rangle$, $\langle 2[shift]$, $2[pow]$, $2[ans]\rangle$, $...$, $\langle (4|n|/|m|-1)[shift]$, $(4|n|/|m|-1)[pow]$, $(4|n|/|m|-1)[ans]\rangle\rangle$, if $s$ is \imped, there is at least one node which is \imped.
    
    This shows that if the global state is \imped, then there exists some node $i$ which is \imped.
    
    Next, we show that if some node is \imped, then node 1 will not store the correct answer.
    $\forall i:i\in[1:4|n|/|m|-1]$ node $i$ is \imped if it has non-updated values in $i[ans]$, $i[pow]$ or $i[shift]$. This implies that the parent of node $i$ also stores incorrect value in its $ans$ variable. Recursively, we have that node 1 stores an incorrect value in $1[ans]$, and thus the global state is \imped.
    %
\end{proof}

{
With the arguments similar to those made in the proof of \Cref{lemma:mul-normal-predicate-self-stabilizing} and \Cref{theorem:algo-mul-normal-algo-ss}, we have the following.
\begin{lemma}
    The predicate
    \begin{center}
        $\forall i:\lnot(\textsc{\Imped-Log-Modulo-Shift}(i)\lor$\\
        $\textsc{\Imped-Log-Modulo-Pow}(i)\lor$\\
        $\textsc{\Imped-Log-Modulo-Ans}(i))$
    \end{center}
    is a lattice-linear self-stabilizing predicate.
\end{lemma}

\begin{theorem}
    \Cref{algorithm:log-modulo} is silent and self-stabilizing.
\end{theorem}
}


\subsubsection{Time complexity analysis}

\Cref{algorithm:log-modulo} is a general algorithm that uses the \textsc{Mod}( \textsc{Mul}($\cdots$)) and \textsc{Mod}( \textsc{Sum}($\cdots$)).
For some given $x,y$ and $z$ values, \textsc{Mod}(\textsc{Mul}($x,y$),$z$) (resp., \textsc{Mod} (\textsc{Sum}($x,y$),$z$)) involves first the multiplication (resp., addition) of two input values $x$ and $y$ and then evaluating the resulting value under modulo $z$.
These functions can be implemented in different ways. Choices for these implementations affect the time complexity. 
We consider the following approaches.

\paragraph{Modulo via Long Division}~\\
First, we consider the standard approach for computing \textsc{Mod}(\textsc{Mul}($\cdots$)) and \textsc{Mod}(\textsc{Sum}($\cdots$)). Observe that in \Cref{algorithm:log-modulo}, if we compute  \textsc{Mod}(\textsc{Mul}$(x,y)$) then $x,y < m$. Hence, we can use Karatsuba's parallelized algorithm from \Cref{section:mul-parallel}, where both the input numbers are less than $m$. Using the analysis from \Cref{section:mul-parallel}, we have that each multiplication operation has a time complexity of $O(|m|)$.

Subsequently, to compute the mod operation, we need to compute $xy \mod m$ where $xy$ is upto $2|m|$ digits long. Using the standard approach of long division, we will need $|m|$ iterations where in each iteration, we need to do a subtraction operation with numbers that are $|m|$ digits long. Hence, the complexity of this approach is $O(|m|^2)$ per modulo operation. Since this complexity is higher than the cost of multiplication, the overall time complexity is $O(|m|^2\times \lg \dfrac{|n|}{|m|})$. 

\paragraph{Modulo by using $M$}~\\
The previous approach used $m$ and $n$ as inputs. Next, we consider the case where $m$ is hardcoded in the algorithm. As discussed in \Cref{section:mod-sequential}, we observe that these problems occur in practice. Our analysis is intended to provide lower bounds on the complexity of the modulo operation when $m$ is hardcoded. Similar to \Cref{subsection:dfa-modulo}, the pre-processing required in these algorithms makes them impractical in practice. However, we present them to show that there is a potential to reduce the complexity by some pre-processing. 

We can use \Cref{algorithm:parallel-karatsuba} for multiplication; each multiplication operation will have a time complexity of $O(|m|)$.
Subsequently, to compute the mod operation, we need to compute $xy \mod m$ where $xy$ is upto $2|m|$ digits long.  Using $M$, we will need $2m$ iterations; each iteration takes a constant amount of time. Hence, the complexity of this approach is $O(|m|)$ per modulo operation. Since this complexity is higher than the cost of multiplication, the overall time complexity is $O(|m|\times \lg \dfrac{|n|}{|m|})$. 

\paragraph{Modulo by Constructing Transition Functions}~\\
In this part, we again consider the case where $m$ is hardcoded.
%
%
If $m$ is fixed, we can create a table $\delta_{sum}$ of size $m\times m$ where an entry at location $(i,j)$ represents $i+j \mod m$ in $O(m^2)$ time. Using $\delta_{sum}$, we can create another transition function $\delta_{mul}$ of size $m\times m$ where an entry at location $(i,j)$ represents $i\times j \mod m$ in $O(m^2)$ time. 

Using a preprocessed $\delta_{mul}$, the time complexity of a \textsc{Mod}(\textsc{Mul}($\cdots$)) operation becomes $O(1)$. As a result, the overall complexity of the modulo operation becomes $O(\lg \dfrac{|n|}{|m|})$. The pre-processing required in this method also is high. However, the effective time complexity of the modulo operation is reduced even more, as compared to the method that uses $M$, which is discussed above.

\section{Discussion on Common Properties of These Algorithms}\label{section:discussion}

In this section, we look at some common properties that are present in the problems and algorithms discussed in the preceding sections. Effectively, we also provide an alternate visualization to the abstraction of the lattices induced by the algorithms present in this article.

\subsection{Data Dependency Among Nodes}

In \Cref{subsubsection:ll-karatsuba}, for example, we showed how \Cref{algorithm:parallel-karatsuba} is lattice-linear by showing that given any suboptimal global state, we can point out specific nodes that are \imped. Any \imped node $i$ has only one choice of action, which implies that a total order is induced among all the local states that $i$ can visit. Such a total order, induced among the local states of every node, gives rise to the induction of a lattice among the global states.

Let that \textit{source} of the variable $i[var]$ in a node $i$ be the node that $i$ depends on to evaluate $i[var]$. For example, under \Cref{algorithm:parallel-karatsuba}, node $i$ depends on node $3i-1$, node $3i$ and node $3i+1$ to evaluate $i[ans]$. Thus the source nodes for node $i$ with respect to the evaluation of $i[ans]$ are node $3i-1$, node $3i$ and node $3i+1$. Similarly the source nodes of $i$ with respect to $i[m]$ or $i[n]$ is node $\Big\lfloor\dfrac{i+1}{3}\Big\rfloor$.

Let that $\textsc{Source}(i,var)$ is the set of nodes that are the source of $i$ with respect to $var$. Thus, under \Cref{algorithm:parallel-karatsuba}, for example, $\textsc{Source}(i,ans)=\{3i-1,3i,3i+1\}$. $\textsc{Source}(i,m)=\Big\{\Big\lfloor\dfrac{i+1}{3}\Big\rfloor\Big\}$. However, $\textsc{Source}(1,m)=\phi$ because $i$ is receiving $m$ as part of the input. Similarly, for all other algorithms, we can define the source nodes for all the nodes with respect to any given variable.

Let $\textsc{Variables}(i)$ be the set of the names of all variables of node $i$. We use a macro $\textsc{Depends}(i)$; we have the following recursive definition for this macro.
\begin{center}
    \begin{tabular}{|l|}
        \hline 
        $\textsc{Depends}(i)\supseteq \bigcup\limits_{var\in \textsc{Variables}(i)}\textsc{Source}(i,var)$.\\
        $\textsc{Depends}(i)\supseteq \bigcup\limits_{j\in\textsc{Depends}(i), var\in \textsc{Variables}(j)}\textsc{Source}(j,var)$.\\
        \hline 
    \end{tabular}
\end{center}

\subsection{Induction of $\prec$-lattice}

In \Cref{algorithm:parallel-karatsuba}, for $m$ and $n$, the information of $m$ and $n$ for $i$ is set based on the values of the parent of $i$. Hence, the $\textsc{Depends}(i)$ will contain all the ancestor nodes of $i$ in the tree. In addition, the $ans$ variable of $i$ is based on the children of $i$. Hence $\textsc{Depends}(i)$ will (also) contain the descendants of $i$ in the tree. For example, in \Cref{figure:100times100}, $\textsc{Depends}(3)=\{1,8,9,10\}$ and $\textsc{Depends}(8)=\{1,3\}$.

Let $\textsc{Is-Bad}(i,var)$ be true if and only if $i$ is impedensable with respect to some variable $var$, i.e., based on the values of the variables in $\textsc{Source}(i,var)$, $i$ has not computed $var$ correctly yet. We define state value of a node $i$ in a global state $s$ as follows. All the macros are also computed in the same global state $s$.
\begin{center}
    $\begin{array}{l}
        \textsc{State-Value}(i,s)=\\
        |\{var|var\in\textsc{Variables}(i):\textsc{Is-Bad}(i,var)\}|\\
        + |\{var|var\in \textsc{Variables}(j): j\in\textsc{Depends}(i,var): \textsc{Is-Bad}(j,var)\}|
    \end{array}$
\end{center}

We define the rank of a global state $s$ as follows.

\begin{equation*}
    \textsc{Rank}(s)=\sum\limits_{\text{each node $i$}}\textsc{State-Value}(i,s).
\end{equation*}


From the perspective of, for example, \Cref{algorithm:parallel-karatsuba}, a total order is induced among the local state visited by a node; $\textsc{State-Value}(i)$ describes the badness of the local state of a node $i$, which decreases monotonously as the nodes execute under \Cref{algorithm:parallel-karatsuba}.
Similarly, for all other algorithms, a total order is defined similarly using $\textsc{State-Value}(i)$.

As a consequence of the total order that is defined by $\textsc{State-Value}(i)$, a lattice among the global states can be observed with respect to the rank of the system; if the rank of a state $s$ is nonzero, then there is some node $i$ that is \imped in $s$. Let that only node $i$ changes its state and as a consequence, $s$ transitions to state $s'$. Then, we have that $s\prec s'$ where $s[i]\prec s'[i]$. This forms a $\prec$-lattice among the global states where $s[i]\prec s'[i]$ iff $\textsc{State-Value}(i,s)>\textsc{State-Value}(i,s')$ and $s\prec s'$ iff $\textsc{Rank}(s)>\textsc{Rank}(s')$. 
Rank is 0 at the supremum of the lattice, which is the optimal state.

From the above observation, we have that the system is able to converge from an arbitrary state to the required state within the expected number of time steps. This allows providing new inputs to a parallel processing system without needing to refresh variables of the nodes.

A problem is lattice-linear if it can be modelled in such a way that an \imped node must change its state in order for the system to reach the optimal state \cite{Garg2020}. 
From the above discussion, we have the following theorem about multiplication and modulo operations.


\section{Related Work}\label{section:literature}

\noindent\textbf{Lattice-linear problems}: The notion of representing problems through a predicate under which the states form a lattice was introduced in \cite{Garg2020}. We call the problems for which such a representation is possible \textit{lattice-linear problems}. Lattice-linear problems were also studied in \cite{Garg2021,Garg2022}, where the authors have studied lattice-linearity in, respectively, housing market problem and several dynamic programming problems. 
Many of these problems do not allow self-stabilization.

{
\noindent\textbf{Non-lattice-linear problems}: Certain problems are non-lattice-linear problems. In such problems, there are instances in which the \imped nodes cannot be distinctly determined, i.e., in those instances
$\exists s :\lnot\mathcal{P}(s) \land (\forall i: \exists s': \mathcal{P}(s')\land s[i]=s'[i]$).
Minimal dominating set (MDS) and several other graph theoretic problems are examples of such problems. 
This can be illustrated through a simple instance of a 2 node connected network with nodes $A$ and $B$, initially both in the dominating set. Here, MDS can be obtained without removing $A$. Thus, $A$ is not \imped. On the same hand, MDS can be obtained without removing $B$, and thus, $B$ is also not \imped.

For non-lattice-linear problems, $\prec$-lattices can induced algorithmically as an \imped node cannot be distinctly determined naturally. Eventually lattice-linear algorithms (introduced in \cite{Gupta2021}) and fully lattice-linear algorithms (introduced in \cite{Gupta2022,Gupta2023a}) were developed for many such problems.

The discussion of non-lattice-linear problems is out of the scope of this paper. However, we discuss it briefly, above, for completeness.
}

\noindent\textbf{Snap-stabilization}: The notion of snap-stabilization was introduced in \cite{Bui1999}.
The algorithms that we study in this article make the system follow a sequence of states that is deterministically predictable because of the underlying lattice structure in the state space. Thus, these algorithms exhibit snap-stabilization. In general, all self-stabilizing algorithms where lattice-linearity is utilized are snap-stabilizing.

\noindent\textbf{Modulo}: In \cite{Zeugmann1992}, the authors have presented parallel processing algorithms for inverse, discrete roots, or a large power modulo a number that has only small prime factors. A hardware circuit implementation for mod is presented in \cite{Butler2011}.

In this article, we study several parallel processing algorithms which are self-stabilizing. Some of these algorithms require critical preprocessing, and some do not. The general algorithm for modulo (\Cref{algorithm:log-modulo}) utilizes the power of (sequential or parallel) modulo and multiplication operations on smaller operands.

\noindent\textbf{Multiplication}: In \cite{Cesari1996}, the authors presented three parallel implementations of the Karatsuba algorithm for long integer multiplication on a distributed memory architecture. Two of the implementations have time complexity of $O(n)$ on $n^{\lg 3}$ processors. The third algorithm has complexity $O(n \lg n)$ on $n$ processors.

In this article, we study lattice-linearity of the first algorithm presented in the above-mentioned article.

\section{Conclusion}\label{section:conclusion}

The contribution of this article is two-fold, one is applicative and the other is mathematical. First, we show that the parallelization of multiplication and modulo is lattice-linear. Due to lattice-linearity, we have that the algorithms we study in this article are tolerant to asynchrony. Second, we show two different lattice structures, for both multiplication and modulo, which guarantee convergence in asynchronous environments. This is the first article that shows that a lattice-linear problem can be solved under two different lattice structures. Specifically, considering (any) one of these problems in both the lattice structures, (1) the numbers of nodes are different, so the size of the global states in both the lattice structures is different, and (2) the numbers of children that a node has are different.

Multiplication and Modulo are among the fundamental mathematical operations. Fast parallel processing algorithms for such operations reduce the execution time of the applications which they are employed in. In this article, we showed that these problems are lattice-linear.
In this context, we studied parallelization of the standard multiplication and a parallelization of Karatsuba's algorithm. In addition, we studied parallel processing algorithms for the modulo operation. 

The presence of lattice-linearity in problems and algorithms allows nodes to execute asynchronously. This is specifically valuable in parallel algorithms where synchronization can be removed as is. 
These algorithms are snap-stabilizing, which means that the state transitions of the system strictly follow its specification.
They are also self-stabilizing, i.e., the supremum states in the lattices induced under the respective predicates are the optimal states. 

Utilizing these algorithms, the available cluster or GPU power can be used to compute the multiplication and modulo operations on big-number inputs. In this case, a synchronization primitive also does not need to be deployed. Also, the circuit does not need to be refreshed before providing it with a new input. This is also very fruitful, for example, in Karatsuba's multiplication the time that it would take to refresh the circuit is $O(|n|^{2\lg 3})$, but even without refreshing, we obtain the final answer in $O(n)$ time. This shows the gravity of the utility of the self-stabilizing property of these algorithms. Thus a plethora of applications will benefit from the observations presented in this article.



\bibliography{modmul.bib}
\bibliographystyle{acm}

\end{document}